\documentclass{roffin}
\usepackage[]{algorithm2e}

\begin{document}

\markboth{Bryce M. Kim}{Zeno machines and Running Turing machine for infinite time}
\author{Bryce M. Kim}
\title{Zeno machines and Running Turing machine for infinite time}

\maketitle

\begin{abstract}
This paper explores and clarifies several issues surrounding Zeno machines and the issue of running a Turing machine for infinite time. Without a minimum hypothetical bound on physical conditions, any magical machine can be created, and therefore, a thesis on the bound is formulated. This paper then proves that the halting problem algorithm for every Turing-recognizable program and every input cannot be devised whatever method is used to exploit infinite running-time of $TM$.
\end{abstract}
\section{Introduction}
Since the discovery of decisions problems that cannot be computed by Turing machines (\cite{turing36}), discussions related to computing what Turing machines cannot has never gone away. For example, in computability theory, there exists discussion of various Turing degrees beyond $0'$(\cite{kleene54}), the degree of halting problem, and how they relate to arithmetic hierarchy. It would therefore be incorrect to state that hypercomputation is a science fiction in logical domain. Most of hypercomputation controversies actually have occurred in physical domain, and often Physical Church-Turing Thesis, which states that scientific laws only allow us to build Turing machine at maximum as a machine, is invoked (\cite{fitz06}) to challenge the possibility of hypercomputation machines. As a constrast, a strong form of Church-Turing thesis is Logical Church-Turing Thesis, which states that it is impossible to build a logically coherent hypercomputation machine model. Logical Church-Turing Thesis is often refuted by invoking the fact that Turing degrees above $0'$ have been researched in computability theory and what Logical Church-Turing Thesis amounts to is the refutation of Turing degrees above $0'$. Since Turing machine is mentioned from the start, it will be beneficial to state the definition of Turing machine, as a reference for the entire paper(\cite{hopcroft79}).
\begin{definition}
Turing machine is a 7-tuple $M= \langle Q, \Gamma, b, \Sigma, \delta, q_0, F \rangle$ with arbitrarily infinite length of a tape $TP$ and the other tape $TP_2$ and the tape head $h$ and $h_2$ for each tape where $Q$ is a non-empty finite set of states that Turing machine can be in, $\Gamma$ is a non-empty finite set of alphabet symbols written to and read off $TP$ and $TP_2$. $b \in  \Gamma$ is the blank symbol, $\Sigma\subseteq\Gamma\setminus\{b\}$ is the set of input symbols, $q_0 \in Q$ is the initial state, $F \subseteq Q$ is the set of accepting states for which Turing machine halt, $\delta: (Q \setminus F) \times \Gamma \times \Gamma \rightarrow Q \times \Gamma \times \Gamma \times \{L,N,R\} \times \{L,N,R\}$ is a transition function by which Turing machine changes its states and writes to the tape where $\delta$ itself is a partial function (it is not a total function) and $L$ represents left shift of $h$ by one tape unit, $N$ represents no shift of $h$, $R$ represents right shift of $h$ by one tape unit.  
\end{definition}
\begin{definition}
Machine $N$ is Turing-equivalent to Turing machine $M$, if $N$ meets the definition of $M$, or if $M$ can simulate $N$ by algorithmic procedure and $N$ can simulate $M$ by its own algorithmic procedure.
\end{definition}
\begin{definition}
Machine $N$ is Turing-complete, if $N$ can simulate execution of instructions written for Turing machine $M$.
\end{definition}
In this definition, any program written for $M$ is understood as the symbols of $TP$ that are read off to change the state of $M$ and act on $h$ and $TP$.
\begin{definition}
Machine $P$ and $N$ are language-equivalent, if they recognize/accept the same set $X$ of inputs and algorithms as the contents of tapes and recognize no more than $X$. If one of the machines is Turing machine $M$ and they are language-equivalent, the corresponding language $L$ is called Turing-recognizable. 
\end{definition}
\begin{definition}
Turing machine $M$ ``computes'' a decision problem, if $M$ reaches accepting or rejecting states for any input for the decision problem in finite time.
\end{definition}
With all above definitions in mind, we may go back into the issue of a decision problem that is not computable by $M$, mainly halting problem. The ordinary diagonalization pseudocode for halting problem goes the following (\cite{penrose90}):\\

Program $i$:\\
\begin{algorithm}[H]
 \eIf{h(i,i) == 0}{
 \KwRet{0} 
 }{
 \While{1}{
 Loop Forever\; 
 }
 }
 \caption{Diagonalization pseudocode}
\end{algorithm}
where $h(x,y)$ represents some arbitrary halting algorithm not specifically defined with $x$ as program and $y$ as inputs to $x$.\\
It is worthwhile to state some misunderstanding that may rise out of the above algorithm and Turing's demonstration of uncomputability of halting problem for Turing-recognizable language. The above diagonalization demonstration only works if $h$ is assumed to Turing-recognizable. That is, if $h$ is not Turing-recognizable, then the whole program $i$ is not Turing-recognizable. In such case, $h(i,i)$ is not solving the halting problem for Turing-recognizable language but solving the halting problem for a non-Turing-recognizable algorithm. A similar assessment can be found in \cite{ord05}.\\
But what if $h$ can be Turing-recognizable and is performed by hypercomputation machine, allowing us to see the outputs of $h$ in finite time in our spacetime reference? Is this logically possible? A such possibility will be explored in the following sections of this paper.
\section{Minimal Physical Computation Thesis}
But before moving onto the construction our machine, we need to set on a minimum criterion on what would constitute as a machine. Without any criterion, machine $P$ may do nothing and some set of spacetimes ``magically'' compute every different problem, when requested by a human being. Therefore, the following Minimal Physical Computation Thesis will be used.
\begin{definition}
Minimal Physical Computation Thesis(MPCT): the logical rules of imaginary or real spacetimes in which machines run must not be dependent on specific time. They may depend on $\Delta t$, or $dt$, infinitesimal. 
\end{definition}
The above definition does not ask whether a machine is physically possible or not according to our laws of physics. Rather, it is a much weaker thesis that asks all machine models to follow at least minimal conditions. With the above thesis, I will discuss Zeno machines.
\section{Zeno machine and Turing machine running for infinity}
Let us first define Zeno machine, as done in \cite{potgieter06}:
\begin{definition}
Zeno machine $ZM$ is a machine that satisfies all criteria for Turing machine except that each execution step of $ZM$ is executed twice as fast as the previous execution step.
\end{definition}
This definition, by itself, asserts that $ZM$ is just a Turing machine except it meddles with time, so it should recognize the same language as Turing machine.\\
But this definition, while not entirely problematic, will suffer from MPCT problems, if not given a specific hypothetical physical environment. It is important to note that the direction of each movement of $h$ in Turing machine will be different. It may shift left, it may shift right. It is well known that Turing machine that can shift only in one direction (often represented by Read-only right moving Turing machine) cannot be a full-featured Turing machine, because Read-only right moving Turing machine is equivalent to DFAs, not full Turing machine (\cite{tucker04}).\\
This suggests that just placing a Turing machine $TM$ in a constant acceleration reference frame will not create a Zeno machine, as head $h$ needs to move in different directions. \\
There is one hypothetical logical possibility that involves visualizing time as we visualize space and solves $MPCT$ problem. Basically, like how space may get contracted, as time goes on, time interval $\Delta t$ gets contract while other space dimensions remain unchanged (in other words, $\Delta x$ remains the same, where $x$ refers to space coordinate). Then when $\Delta t = 0$ (or $st(\Delta t) = 0$ where $st$ refers to standard part function), spacetime refreshes itself, bringing us back into some initial $\Delta t = \mu_0$, where $\mu_0$ is distance-like measure, and then continues contracting to $\Delta t = 0$ and going back to $\mu_0$ again. As the physical rule itself is uniform across time, this does not violate $MPCT$.\\    
While the example seem impossible according to our understanding of physics, this will not hamper with our results below, because the paper is concerned with the logical development from minimal physical conditions, not physical plausibility.\\
It is important to note that Zeno machine is just a $TM$ that runs for infinite time, with the help of physical environment, Zeno machine will only accept Turing-recognizable algorithms. And if Zeno machine can solve the halting problem for every Turing-recognizable algorithms and inputs, and $ZM$ is a $TM$, then a Zeno algorithm that solves the halting problem would be Turing-recognizable, if such one exists.\\
Whether Zeno machine ``solves'' the halting problem by an algorithm is debatable. The usual way of solving the halting problem in $ZM$ is to simulate each execution step of a Turing machine and a program, check whether the program halted, modify and mark into a separate tape of $ZM$ to show whether the program halted or not. When $ZM$ stops in finite time, running infinite number of executions, it did not really give an output. What we see as an output is rather a result of $ZM$ stopping. If $0$ represents not halted and $1$ represents halted, then no presence of $1$ on a separate tape implies that the program never halted. In this interpretation, there exists no algorithm for $ZM$ that solves the halting problem, because $ZM$ by itself cannot determine the output. However, the algorithm that ``tries'' (but does not give an output if the program never halts) to solve the halting problem is Turing-recognizable.\\
There is the other way of solving the halting problem in $ZM$, which is more algorithmic. The below halting problem algorithm for $ZM$ is basically a modification of the first method, except that this time $ZM$ also performs division of $1$ by half each time it is called.\\
Let the tape alphabet symbols of the below Zeno machine be defined by $\{0,1,2,3,4,5,6,7,8,9, b\}$, where $b$ is the blank symbol.\\\\
Algorithm 2 starts:\\
\begin{algorithm}[H]
Receive an input that is a program $p$\;
$x=1$, and $x$ is stored in a separate tape $TP_x$ of $ZM$\;
\While{1}{
Execute a single succeeding instruction/execution step of $p$. If the program halts, go to the accepting/halting state as required, and escape the while loop.\;
$x=x/2$ and instead of writing on the original tape locations of $x$, it is fine to write the new $x$ after the original $x$, separated by $b$\;
\If{Last digit of $x$ is zero, as pointed by the head of $TP_x$}
{Move $ZM$ to a unique state that frees the heads of $ZM$ from being stopped. After this state transition, $ZM$ starts in the previous initial speed of $ZM$ before the execution of Algorithm 2. Escape the while loop.
}
}
Check whether $ZM$ is in the halting state - if so, output 1. if not, output 0.\;
\caption{Halting problem algorithm for $ZM$}
\end{algorithm}

In practice, $ZM$ will check whether the program $p$ halted in constant time $t_p$, which is 2 seconds if the initial execution step takes 1 second and each succeeding step takes half of the previous step's execution time. The reason for extra ``If'' conditional is to allow our halting procedure to be modelled as Turing-recognizable algorithm involving state transitions.\\
But now the question: how can we be sure that the infinite division of $1$ by half will result in the last digit being zero?\\
\begin{proposition}
It is impossible for a $TM$ in any spacetime or any reference frame-reference reference relationship to have non-zero digit as the result of the infinite division of $1$ by half.
\end{proposition}
\begin{proof}
To specify what division would do in $TM$, let us choose base-2 notation. That is, the first digit of binary number is understood as being $2^0$, the next digit, $2^{-1}$, the succeeding digit, $2^{-2}$ and so on. $0.5$ will be $01$ in base-2 notation, $0.25$ will be $001$ and so on. Therefore, what division by half does is shifting the number by one to the right.\\
Each writing operation of one digit occupies one cell - therefore, the location number, or ordinal, of the rightmost non-empty part of the tape increases by $1$. By both ZF set-theoretical argument following from the Axiom of Infinity, or by the intuitive understanding that finite number cannot suddenly jump to infinite number, the first infinite ordinal, $\omega$ exists.\\
Suppose that as the result of the infinite division of $1$, the last digit proves to be $1$. Let the first infinite division result be $I$. $I$ has infinite number of $0$'s before $1$. The number of divisions that occurred therefore is $\omega$. But at the $\omega - 1$ division of $1$ by half, there will still be infinite number of $0$'s before $1$, because infinite number of $0$ cannot become finite number of $0$ just by taking away one $0$.\\
Note that in Zermelo-Fraenkel set theory, $\omega -1$ does not exist as an ordinal number. This is possible because in mathematical setting, this type of physical problem does not exist.\\
Therefore, the infinite division of $1$ cannot result in the last digit of the number being non-zero.
\end{proof}
\begin{proposition}
Either the infinite division of $1$ is not computable even in infinite time, or the division results in zero.
\end{proposition}
\begin{proof}
This is the tautological statement derived from the preceding proposition.
\end{proof}
\begin{proposition}
Let us assume, to the contrary of the propositions above, that $TM$ does have non-zero last digit as the infinite divisions of $1$ by half. Then $TM$ can recognize that it has run for infinite time, because $\omega+1$th division result equals $\omega$th division result.
\end{proposition}
\begin{proof}
If $TM$ does have non-zero last digit as the infinite divisions of $1$, then this suggests that $\omega - 1$th division is not accessible or does not exist. Therefore, $TM$ will no longer be $TM$. But suppose that our assertions are wrong so that $TM$ is still $TM$.\\
As $\omega -1$th division is not accessible, $\omega$th division has the set of $0$'s before the last digit that is of order type $\omega$, which means that each tape cell of the $\omega$th division result before the last digit cell can be placed in bijection with $\mathbb{N}$. Adding one additional $0$ to the left of the cells of the result does not change the order type. Therefore, when $\omega+1$th division and $\omega$th division are compared, they should turn out to be equal.\\
Just add the algorithm that compares the two numbers from the first digit to the last digit digit-by-digit from the leftmost side, and $TM$ can know whether it is in $\omega+1$th division step. 
\end{proof}
Another question: would not the answer of $0$ as the infinite division of $1$ by half cause contradictions? $0 \times 2 = 0$ and because the first infinite division of $1$ by half is $0$, $\omega -1$th division of $1$ by half may also need to be $0$.\\
This is not actually the case, because it assumes that shifting left or multiplication is an inverse operation to shifting right or division. The generalization of the finite case where multiplication is inverse to division does not occur at the infinite level. Replacing $0$ with infinitesimal number cannot change this conclusion, because $\omega -1$th division of $1$ would still need to be infinitesimal.\\
Assuming that the infinite division of $1$ can be computed in infinite time (and space), during the computation of the infinite division, infinite number of $0$ may appear, and if $\omega$th division shows such a behaviour, would not this also cause contradiction?\\
But it is very easy to see that this question is actually equivalent to the question of whether $0 \times 0$ implies the contradiction. Also, in such case infinite number of digits of $0$ can be truncated, leaving only one $0$ before moving onto the next step.\\
What if the infinite division of $1$ is uncomputable? This causes no problem for $ZM$. The infinite division of $1$ resulting in $0$ in $ZM$ can also be understood as $0$ referring to the state of undefined. The rationale for this can be inferred by $ZM$, because it is easy to prove in finite time in ordinary $TM$ that finite number of divisions of $1$ by half cannot be zero.\\
It is also crucial to note that infinite division of $1$ not computable in a certain machine operation even with infinite time does not mean that infinite division of $1$ is undefined. Therefore, one cannot say that because ordinary $TM$ running for infinite time may not compute infinite division of $1$, infinite division of $1$ being defined in different machine operation modes is wrong. It is easier to understand this argument, if ``not computable'' is replaced by ``execution steps having not halted and continuing to run.''\\
From a different perspective, the division-by-half part of our $ZM$ algorithm above can be understood as taking the standard part of the division result result. As such, infinite division of $1$ defined as zero is not problematic.\\
It is, for sure, weird to say that the calculation of infinite division of $1$, either understood as standard part function of the divisions or normal divisions, will result in zero and seems to cause problems. But at least this does not create an outright contradiction, while the non-zero last digit of the infinite division result causes an outright contradiction.\\
Until now, I have assumed that uncomputability possibility exists.
\begin{proposition}
It was demonstrated above that it is impossible to read infinite number of cells if the infinite division operations of $1$ by half is not computable in infinite time. Assume that it is indeed impossible to access all cells in the tapes of $TM$. Then the number of cells accessible in $TM$ is finitely bounded.   
\end{proposition}
\begin{proof}
There is really nothing to prove here. If the number of cells is not infinite, it must be finite. When the number of cells is finite, that number is constant $k \in \mathbb{N}$. If not constant $k$, there is no proof that the number of cells is finite.
\end{proof}
\begin{proposition}
Universal Turing machine that simulates every other Turing machine cannot have finite number of cells in its tapes. 
\end{proposition}
\begin{proof}
A simple Turing-recognizable program that writes $1$ to a new tape cell each time would suffice. 
\end{proof}
Therefore, it can be concluded that infinite division of $1$ is computable, given infinite time. 
\begin{proposition}
If infinite division of $1$ by half results in, or is, $0$, then it is assured that the above algorithm will give an output for halting problem of Turing machines. 
\end{proposition}
\begin{proof}
Suppose that $TP_x$ contains non-zero digit as the last digit of the division, but the program executed the infinite number of execution steps that make the heads of the tapes of $ZM$ not move. But this means contradiction, because the fact that $TP_x$ contains non-zero digit as the last digit implies that there have only been finite number of executions plus some finite number of succeeding executions. Therefore, it must be the case that $TP_x$ contains $0$ when the heads of $ZM$ stop for the given moments.
\end{proof}
Because Algorithm 2 is Turing-recognizable, the existence of Algorithm 2 that solves the halting problem leads to the following diagonalization paradox:
For simplicity of proof, it will be assumed that our universal Turing machine stops every program runs in a clocked fashion led by division-of-1-counter. So every program runs one execution step, asks a counter divide by half, waits until the counter finishes the division. Also, after reaching $\omega$ state, the counter clears up infinite number of $0$'s. (So, if infinite division of $1$ results in $0$, as it is argued in the paper, then $0$ will be the counter's result afterwards, or otherwise, the counter basically repeats infinite division of $1$ again)
\begin{lemma}
The running time of the Algorithm 2 is $O(\omega \cdot 2) = O(\omega + \omega)$, as the first $\omega$ comes from running programs, $\omega$ comes from computing of $\omega$ number of $0$'s in division of $1$. But if the assumption above is used, it reduces to $O(\omega)$ in terms of program's own non-counter execution time.
\end{lemma}
\begin{definition}
$O(\omega)$ is defined as running time that necessarily runs for $\omega$ and may run for additional some finite time. Similarly, $O(\omega + \omega)$ means the running time of $\omega+j+\omega+k = \omega+\omega+k$ where $j,k \in \mathbb{N}$. 
\end{definition}
\begin{proposition}
There exists an algorithm that checks whether a program $p$ with input set $i$ has the following if Algorithm 2 that checks whether $p$ with $i$ halts in finite time exists (As said above, Algorithm 2 runs for $O(\omega)$-time.):\\
Either:\\
1. Program stops before $\omega$-time\\
OR all of:\\
2-1. there exists `if(R(f(x),k))' instruction for $p$ with $i$ that is initiated at finite time but only stop at $\omega$-time, where $f$ is some Turing-recognizable algorithm/function and $x$ is its input set and $k \in \Gamma^n$ where $n \in \mathbb{N}$ and $R$ is some relation,\\
2-2. the `if(R(f(x),k))' that 2-A-1 applies runs only for at most $w \in N$-time after reaching $\omega$-time, 
2-B. $p$ with $i$ stops before $\omega \cdot 2$-time.\\
2-1 and 2-2 must both be satisfied. If either condition 1 or 2 is satisfied then the algorithm prints out $1$. Otherwise, print $0$.     
\end{proposition}
For simplification of our proof, it will be assumed that all `if' instruction has the similar machine code that can be checked by a finite-time running algorithm. It is soon seen that this does not cause any problem for the proposition.
\begin{proof}
We use a dove-tailing approach.\\
Let each tape divided by sub-tape section that is used by single $m$. Different sub-tape sections are distinguished by a special character $\nu$. each sub-tape section is divided into sub-areas that are used by single $i$. Different sub-areas are distinguished by a special character $\mu$. The proof will be given by the combination of Algorithm 3, 4 and 5.\\  
\begin{algorithm}
Receive input $p$ and $i$\;
$r == 1$\;
$m == 1$\;
$t == 0$\;
$q == 0$\;
 \While{Division of $1$ by half is yet to reach zero}{
 \For{$r \leq m$}{
 Check whether $r$ is the newly created sub-area by checking $m-1$th sub-tape. If so, create a new sub-area in the $m$th sub-tape\;
 Continue running a single instruction from computations left from sub-area $r$ (or the parent sub-area of $r$, if $r$ is a newly-created sub-area) of sub-tape $m-1$. If necessary, also write down the location of instructions being read and the location of the next instruction to be executed\;
 Check whether succeeding instructions to be executed right after executing $b$, when running $p$ with $i$, are the form of $if(R(f(x),k)$ where $if$ a usual if branch conditional, $R$ is relation of $==$, $\leq$, $\geq$, $<$ or $>$, $f$ is a Turing-recognizable algorithm/function, $k \in \Gamma^n$ where $n \in \mathbb{N}$. If it does, mark at the end of the sub-area that two new sub-areas need to be created at the next $m$. Two new sub-areas are created, because each sub-area assumes different value (true or false) to relation $R$ and carry out computation from there. The newest sub-area is numbered after the number of the last sub-area\; 
 Check whether there exists the set $M$ of sub-areas that can be labelled as ``killed'' because of the result reached in sub-area $r$ (that is, $if(R(f,k),k)$, that made some $j \in M$th sub-area to separate as a new sub-area, is determined in the branch related to sub-area $i$.), and if so, fill every $j \in M$th sub-area in previous sub-tapes with special ``kill'' character $\Phi$\;
 Check whether $r$th sub-area in previous sub-tapes is marked with $\Phi$. If so, also mark a single cell of $i$th sub-area with $\Phi$\;
 If the branch that is associated with sub-area $r$ halted, mark with halting character $\pi$\;
 Whenever $\Phi$ and $\pi$ are both applicable, $\Phi$ rules over $\pi$\;
 If $r$th sub-area does not exist, mark with special character $psi$\;
 Mark the cell next to the last cell of sub-area $r$ with special character $\mu$.
 \If{$r == 1$}{
 Check whether $p$ with $i$ halted. If so, $t == 1$ and escape the While loop\;
 }    
 }
 Check whether there exists a sub-area not filled with $\Phi$, if so $q==1$\;
 Mark with special character $\nu$\;
 $m = m+2^m$\; 
 }
\caption{Demonstrating infinity-time diagonalization paradox: program $u$}
\end{algorithm}
\begin{algorithm}
 \eIf{$t==1$}{
 \KwRet 1\;
 }
 {
 \If{$q==1$}{\KwRet 0\;}
 Check the last sub-tape and find the second and third sub-areas that are not filled with $\Phi$\;
 \If{the second sub-area, not filled with $\Phi$, is filled with $\pi$}
 {Run next execution steps of $p$ with $i$ for time $w$ for every $p$ and $i$ that is not $u$, Algorithm 5. If the result has implication for the second and third sub-areas, kill a sub-area that is inconsistent with the result, and check whether a consistent sub-area halted (That is, proper `If' result is matched). If so, \KwRet 1\; Check also whether $y == p$ and $y == i$, where $y$ is Algorithm 5. If so, first set $u(y,y)==1$ and check consistency of both sub-areas. If both result in inconsistency, check $u(y,y)==0$. If both are inconsistent, then $y(y)$ is undefined and \KwRet 0\; If $u(y,y)==1$ is consistent, while $u(y,y)==0$ is not, check whether second sub-area halted. If so, \KwRet 1\; If not, \KwRet 0\; Similarly, if $u(y,y)==0$ is consistent, while $u(y,y)==1$ is not, check whether third sub-area halted. If so, \KwRet 1\; If not, \KwRet 0\; Otherwise, \KwRet 0\;}
 }
 \caption{Continuation of program $u$}
\end{algorithm}
\begin{algorithm}
Receive input $i$\;
\eIf{u(i,i) == 0}{\KwRet 1\;}
{Some code that can never stop\;}
\caption{Program $y$}
\end{algorithm}
Look at Algorithm 5, or program $y$. If $u(y,y) == 0$, then it stops before $\omega \cdot 2$ and $y(y)$ is defined before $\omega \cdot 2$, because $u$ stops in $O(\omega)$-time, but $y(y)$ either should not stop at $\omega \cdot 2$ or $y(y)$ is undefined. Contradiction. If $u(y,y) == 1$, then it does not stop before $\omega \cdot 2$. Contradiction.\\
It should be noted that this proof is possible, because we can hard-code the existence of $y$ and already figure out that the only `if' loop of $y(y)$ will be $u(y,y)$, repeating endlessly. Therefore, one can know exactly which sub-area to look for, avoiding the conclusion that $u(y,y) == 0$ can be justified because $u$ has `if' conditional that should run more than $w$-time, but still finite time, to finish. 
\end{proof}
Because the ability to check halting problem leads to algorithm 3,4 and 5, halting problem is not solvable by $TM$ even in infinity.\\
It can be argued that because specific indexes are used for `If' instruction and program $y$, the result cannot be valid, because such a method results in proof of uncomputability of halting problem to break down. This is clearly not the case. The reason why halting problem algorithm $h$ cannot just check paradoxical $h(i,i)$ in Algorithm 1 is the fact that $h$'s input space is every Turing-recognizable program and input. Furthermore, there is the set of infinitely many indices for program $i$ that may not be recursive (it is indeed not recursive). Therefore, if one algorithm $p$ creates a paradox for $h$, $h$ is not computable (recursive). Our algorithm 3,4,5 are different from halting problem algorithm in this context.\\ 
This implies that even using $ZM$, halting problem for every Turing algorithm and input cannot be solved.\\
Note that Algorithm 2 can be extended to an ordinary Turing machine running for infinite time, if infinite divisions of $1$ by half results in the last digit being $0$, as we have concluded. It is just that, as finite-living beings, humans cannot see the results in finite time with some inputs, while $ZM$ allows humans to see the results in finite time. 
\begin{proposition}
It is impossible to solve the halting problem for Turing machines by utilizing a Turing machine and infinite time in any form. If the halting is solved using a Turing machine and infinite time, the extra solving power does not from the Turing machine. A finite time case is already proven, so there is no need for proof.
\end{proposition}
\begin{proof}
Run infinite divisions of $1$ by half in ordinary $TM$ and $ZM$ in infinite time (In $ZM$'s case, it will be finite time in our reference frame). But in such case, because Algorithm 2 leads to the diagonalization paradox seen in Algorithm 3 and 4, the halting problem for Turing machines cannot be computed using either $ZM$ or ordinary $TM$ and infinite time. If any Super-Turing algorithm for Super-Turing machine can be reduced to the above Turing-recognizable algorithm, then such a Super-Turing algorithm does not work.   
\end{proof}
This leads to the paradoxical state that whether $TM$ is in the infinite state is not checkable, while we just proved that such a checker exists, by dividing $1$ by half. Denying the existence of infinite tape is not desirable, because as restated in the paper, then Universal Turing machine is no longer universal. How this paradox will be resolved will not be discussed.
\section{Conclusion}
Throughout past years, it has well been recognized that if halting problem algorithm $h(p,i)$, where $p$ is a Turing-recognizable program and $i$ is $p$'s input set, is computable/recursive, then it creates a diagonalization paradox. What has not been often recognized, though, is the fact that diagonalization paradox also works when $h$ is assumed to just be Turing-recognizable. This paper studies the possibility of building a Turing-recognizable algorithm using a Zeno machine and possible hypothetical physical environment that surrounds the machine. How a Zeno machine connects to running a Turing machine for infinite time is also studied. It is then concluded that while the algorithm can be created, it obviously results in the diagonalization paradox, and therefore cannot compute halting problem. This also leads to the problematic state that the use of infinite tape for Universal Turing machine may also be problematic. How these paradoxes will live with empirically justified usages of Turing machines will be left as future works.

\end{document}